\documentclass[12pt]{article}
\usepackage{amsmath,amssymb,amsthm,booktabs}
\usepackage{mathrsfs}
\usepackage{geometry}
\usepackage[T1]{fontenc}
\usepackage[utf8]{inputenc}
\usepackage[colorlinks,citecolor=blue,urlcolor=blue]{hyperref}
\usepackage{newtxtext,newtxmath}
\usepackage{indentfirst}
\usepackage{bm}
\usepackage{graphicx}
\usepackage[round]{natbib}
\numberwithin{equation}{section}
\theoremstyle{plain}
\newtheorem{theorem}{Theorem}[section]
\newtheorem{lemma}{Lemma}[section]
\newtheorem{corollary}{Corollary}[section]

\theoremstyle{definition}
\newtheorem{definition}{Definition}[section]

\newtheorem{example}{Example}[section]
\newtheorem{remark}{Remark}[section]
\usepackage{setspace}
\begin{document}

\title{Discrete asset pricing under transaction costs and model uncertainty with and without short-sale constraints}
\author{Wenqing Zhang\thanks{College of Mathematics and Physics, China Three Gorges University, Yichang, Hubei, China. Email: wqzhang1025@163.com.}}
\date{}
\maketitle

\begin{abstract}
	We study discrete-time asset pricing with bid-ask spreads and model uncertainty. The family of probability measures enters the no-arbitrage condition through the union of its supports. In the single-period setting, we establish fundamental theorems of asset pricing with and without short-sale constraints. In the unconstrained market, no arbitrage is equivalent to the existence of a full-support martingale consistent price system. Under short-sale constraints, the martingale condition is replaced by a supermartingale condition. We then extend these results to a finite multi-period tree. The initial information is allowed to be nontrivial, so the initial trading cost and valuation bounds may depend on the initial state, and the corresponding inequalities are formulated conditionally.
	Finally, for a family of pricing measures, we introduce lower and upper robust supermartingale consistent price systems. We show that no arbitrage implies the existence of a lower system, while the existence of an upper system is sufficient for no arbitrage. A two-state example shows that the lower condition alone is not sufficient.
\end{abstract}


\noindent\textbf{Keywords:} Asset pricing; Transaction costs; Model uncertainty; Short-sale constraints; Finite-state market

%

\section{Introduction}
\label{sec:introduction}

The classical theory of asset pricing is built on the principle that an arbitrage-free market admits a suitable pricing rule. In a frictionless and unconstrained market, this rule is typically represented by an equivalent martingale measure. Bid--ask spreads, short-sale constraints, and uncertainty about the probability probability law each alter this dual description.

With transaction costs, a risky asset is purchased at the ask price and sold at the bid price. A standard dual object is therefore a shadow price process inside the bid-ask spread. If the shadow price is a martingale under an equivalent probability measure, the pair is a consistent price system. This approach goes back to \citet{Jouini1995}; finite-state and discrete-time formulations are developed, among others, by \citet{Kabanov01}, \citet{Schachermayer04}, \citet{Follmer16}, and \citet{Bielecki15}. Trading at a shadow price is never less favorable to the trader than execution at the bid or ask, while the martingale property rules out a systematic gain under the pricing measure.

Model uncertainty replaces a single probability measure by a family of possible measures. Nondominated formulations require genuinely quasi-sure arguments; see \citet{DM06}, \citet{Bouchard15}, \citet{BouchardNutz16}, and \citet{Bayraktar16}. Model-free and pointwise approaches are studied by \citet{Acciaio16}, \citet{Burzoni16}, and \citet{Burzoni19}, and \citet{Obloj21} relate the pathwise and quasi-sure viewpoints. Nonlinear expectation theory provides another language for families of measures; see \citet{Peng2004} and \citet{Peng2019}.

The present finite-state framework is more specific. We assume that every state or path is assigned positive probability by at least one probability model. This condition makes the union of the probability supports equal to the prescribed state space. Since the state space is finite, a finite convex combination of physical measures yields a probability measure with full support.
Thus the family of probability measures affects the no-arbitrage definition through its support, but the model is not nondominated in the technical sense. Making this reduction explicit prevents the probability family from being confused with the auxiliary family of pricing measures introduced later.

Short-sale constraints form another important market friction. They may arise from regulation, borrowing costs, margin requirements, or a lack of lendable securities. Empirical evidence is discussed by \citet{Battalio11}. The dual effect is well known: restricting risky holdings to be nonnegative replaces a martingale pricing rule by a supermartingale rule; see \citet{JouiniShort95}, \citet{Pulido14}, and \citet{Coculescu19}. Related implications for derivative and futures markets are considered by \citet{Jarrow15} and \citet{He20}.

This paper extends the frictionless finite-state analysis of \citet{YangZhangAsset} to markets with bid-ask spreads. The single-period argument uses finite-dimensional separation to construct a full-support consistent price system. In the multi-period model, the same conclusion does not follow by independently choosing a one-period pricing measure at every date. We instead formulate the market through nodewise solvency cones and invoke the finite-state fundamental theorem of asset pricing on each initial subtree. This also clarifies the normalization required when $\mathcal F_0$ is nontrivial.

Under short-sale constraints, the dual risky components become supermartingales. We then compare lower and upper conditions relative to a family $\mathcal Q$ of pricing measures. The lower condition follows from the ordinary constrained theorem by taking a singleton family; it is therefore an existence corollary, not a new equivalent characterization. The upper condition is stronger and is sufficient for no arbitrage. An explicit two-state example shows that a lower system satisfying the support condition can coexist with an arbitrage.

The remainder of the paper is organized as follows. Section \ref{sec:single} treats the single-period market and gives the counterexample separating the lower and upper conditions. Section \ref{sec:multi} develops the multi-period theory, including the nodewise cone construction, conditional dual bounds, and the constrained supermartingale formulation. Section \ref{sec:conclusion} summarizes the scope and limitations of the results.

\section{Single-period asset pricing}
\label{sec:single}

\subsection{Basic financial market model}

We begin with a one-period market in order to isolate the effects of bid-ask spreads and model uncertainty before introducing dynamic trading. At time $0$, the investor chooses a portfolio and holds it until time $1$, when all risky positions are liquidated. The riskless asset is used as the numeraire, so all quantities below are discounted. This formulation makes the distinction between the cost of establishing a position and its terminal liquidation value explicit.

Let
\[
\Omega=\{\omega_1,\ldots,\omega_K\}
\]
be a finite state space. The market contains one riskless asset and $M$ risky assets. The riskless asset is perfectly liquid and is used as the numeraire, so
\[
\underline S_0(t)=\overline S_0(t)=1,
\qquad t=0,1.
\]
For each risky asset $m=1,\ldots,M$, let $\underline S_m(t)$ and $\overline S_m(t)$ denote the discounted bid and ask prices, respectively. We assume
\begin{equation}
\label{eq:single-spread}
0\le \underline S_m(t)\le \overline S_m(t),
\qquad t=0,1,
\quad m=1,\ldots,M.
\end{equation}
The bid-ask formulation includes the proportional transaction-cost model as a special case: if $S_m(t)$ is a reference frictionless price and $\lambda_m(t)\in[0,1]$, then
\[
\underline S_m(t)=(1-\lambda_m(t))S_m(t),
\qquad
\overline S_m(t)=(1+\lambda_m(t))S_m(t).
\]

The initial prices are deterministic, while the terminal prices are random variables on $\Omega$. Model uncertainty is described by a nonempty family $\mathcal P$ of probability measures on $\Omega$. We assume
\begin{equation}
\label{eq:single-P-support}
\sup_{P\in\mathcal P}P(\{\omega\})>0,
\qquad \forall\omega\in\Omega.
\end{equation}
Thus every state is possible under at least one probability model.

\begin{lemma}[Support reduction]
\label{lem:single-support-reduction}
Under \eqref{eq:single-P-support}, there exists a probability measure
$R_{\mathcal P}\in\operatorname{conv}(\mathcal P)$ with full support such that, for every
nonnegative random variable $X$,
\[
\sup_{P\in\mathcal P}E_P[X]>0
\quad\Longleftrightarrow\quad
E_{R_{\mathcal P}}[X]>0
\quad\Longleftrightarrow\quad
X(\omega)>0\ \text{for some }\omega\in\Omega.
\]
Consequently, the no-arbitrage condition in this finite-state model depends
on $\mathcal P$ only through the union of its supports.
\end{lemma}

\begin{proof}
For every $\omega_k$, choose $P_k\in\mathcal P$ with
$P_k(\{\omega_k\})>0$ and set
\[
R_{\mathcal P}=\frac1K\sum_{k=1}^K P_k.
\]
Then $R_{\mathcal P}$ has full support.  For $X\ge0$, each of the three
conditions is equivalent to strict positivity of $X$ in at least one state.
\end{proof}

A trading strategy is represented by its net holdings
\[
h=(h_0,h_1,\ldots,h_M)\in\mathbb R^{M+1}.
\]
Because simultaneously opening offsetting long and short positions only
pays the spread without changing the net payoff, it is without loss to use
the minimal cost of establishing $h$.  The initial discounted cost is
\begin{equation}
\label{eq:single-C0}
C(0)
=
h_0+
\sum_{m=1}^M
\left[
h_m^+\overline S_m(0)
-
h_m^-\underline S_m(0)
\right].
\end{equation}
The terminal discounted liquidation value is
\begin{equation}
\label{eq:single-V1}
V(1)
=
h_0+
\sum_{m=1}^M
\left[
h_m^+\underline S_m(1)
-
h_m^-\overline S_m(1)
\right],
\end{equation}
where $h_m^+=\max\{h_m,0\}$ and $h_m^-=\max\{-h_m,0\}$. The asymmetry between \eqref{eq:single-C0} and \eqref{eq:single-V1} is caused by transaction costs: a long position is opened at the ask price and liquidated at the bid price, whereas a short position generates the bid price initially but must be covered at the ask price at maturity. This asymmetry is the reason that a single frictionless asset price is replaced below by a shadow price selected from the bid-ask interval.

\subsection{Asset pricing under transaction costs and model uncertainty}

We now relate the absence of arbitrage to the existence of a linear pricing rule. The argument proceeds in three steps. We first define arbitrage in a way that reflects the family of possible probability models. We then introduce a consistent price system, which consists of a full-support pricing measure and a shadow price inside the spread. Finally, we prove that these two objects are equivalent in the finite-state market and use the resulting family of price systems to obtain valuation bounds.

\begin{definition}[Arbitrage under model uncertainty]
\label{def:single-arbitrage}
A trading strategy $h$ is called an arbitrage opportunity under model uncertainty if

\emph{(i)} $C(0)\le0$;

\emph{(ii)} $V(1;\omega)\ge0$ for all $\omega\in\Omega$, and
\[
\sup_{P\in\mathcal P}E_P[V(1)]>0.
\]
The market is said to satisfy $NA(\mathcal P)$ if no such strategy exists.
\end{definition}

\begin{remark}
Under \eqref{eq:single-P-support}, condition \emph{(ii)} is equivalent to
\[
V(1;\omega)\ge0\quad\forall\omega\in\Omega,
\qquad
V(1;\omega_0)>0\quad\text{for some }\omega_0\in\Omega.
\]
Indeed, strict positivity at one state is detected by at least one measure in $\mathcal P$.
\end{remark}

The preceding definition uses the family $\mathcal P$ only to determine whether a nonnegative terminal payoff is genuinely positive in at least one economically possible state. The pricing measure introduced next has a different role. It is not selected as the true probability model; instead, it supports a linear valuation under which a shadow price inside the bid-ask spread has the martingale property.

\begin{definition}[Consistent price system]
\label{def:single-CPS}
A pair $(Q,S^*)$ is called a consistent price system if $Q$ is a probability measure satisfying
\[
Q(\{\omega\})>0,
\qquad \forall\omega\in\Omega,
\]
and, for each risky asset $m=1,\ldots,M$,
\[
S_m^*(t)\in[\underline S_m(t),\overline S_m(t)],
\qquad t=0,1,
\]
and
\[
S_m^*(0)=E_Q[S_m^*(1)].
\]
\end{definition}

Denote by
\begin{equation}
\label{eq:single-M}
\mathcal M
:=
\{(Q,S^*):(Q,S^*)\text{ is a consistent price system}\}
\end{equation}
the set of all single-period consistent price systems. Introducing the set $\mathcal M$ is useful because the pricing system need not be unique. Transaction costs generate an interval of admissible shadow prices, and model uncertainty does not single out a unique pricing measure. Thus the natural dual object is a family of admissible pairs rather than one distinguished risk-neutral measure.

We can now state the single-period fundamental theorem. The ``if'' direction follows from comparing the actual trading cost with the value of the same portfolio at the shadow price. The converse direction uses separation of the attainable cone from the positive orthant and then reconstructs both the probability weights and the shadow prices from the separating functional.

\begin{theorem}[Single-period fundamental theorem of asset pricing]
\label{thm:single-FTAP}
The market satisfies $NA(\mathcal P)$ if and only if there exists a consistent price system. Equivalently,
\[
NA(\mathcal P)
\quad\Longleftrightarrow\quad
\mathcal M\ne\varnothing.
\]
\end{theorem}

\begin{proof}
Assume first that $(Q,S^*)\in\mathcal M$. Since $S_m^*(0)$ lies inside the initial spread,
\[
C(0)
\ge
h_0+\sum_{m=1}^M h_mS_m^*(0).
\]
Using the martingale equality and the terminal spread inequalities gives
\begin{align*}
C(0)
&\ge
E_Q\left[h_0+\sum_{m=1}^M h_mS_m^*(1)\right]\\
&\ge
E_Q\left[h_0+\sum_{m=1}^M
\bigl(h_m^+\underline S_m(1)-h_m^-\overline S_m(1)\bigr)
\right]
=E_Q[V(1)].
\end{align*}
If an arbitrage opportunity existed, then $V(1)$ would be nonnegative in every state and strictly positive in at least one state. Since $Q$ has full support, $E_Q[V(1)]>0$, contradicting $C(0)\le0$.

Conversely, assume $NA(\mathcal P)$. Let
\[
\mathcal G
=
\left\{
(-C(0),V(1;\omega_1),\ldots,V(1;\omega_K)):
h\in\mathbb R^{M+1}
\right\}
\]
and define the free-disposal cone
\[
\mathfrak C=\mathcal G-\mathbb R_+^{K+1}.
\]
For any two net portfolios $h$ and $g$, the bid-ask inequalities give
\[
C(h+g)\le C(h)+C(g),
\qquad
V(h+g;\omega)\ge V(h;\omega)+V(g;\omega).
\]
Thus netting opposite positions weakly improves every coordinate of
$(-C,V)$.  Equivalently, $\mathfrak C$ is the cone generated by the
elementary long and short positions, the two cash directions, and the
free-disposal directions.  It is therefore a closed polyhedral convex
cone.  The no-arbitrage condition implies
\[
\mathfrak C\cap\mathbb R_+^{K+1}=\{0\}.
\]
Indeed, if only the initial-cost coordinate of a nonzero intersection
point were positive, the initial receipt could be retained in the numeraire
to produce a strictly positive terminal payoff.
Strong separation therefore yields a vector
\[
f=(f_0,f_1,\ldots,f_K)\in(0,\infty)^{K+1}
\]
such that $f\cdot y\le0$ for every $y\in\mathfrak C$. Applying this inequality to attainable claims gives
\begin{equation}
\label{eq:single-separation}
C(0)
\ge
\sum_{k=1}^K\frac{f_k}{f_0}V(1;\omega_k).
\end{equation}
Applying \eqref{eq:single-separation} to one unit and minus one unit of the riskless asset shows
\[
\sum_{k=1}^K\frac{f_k}{f_0}=1.
\]
Hence
\[
Q(\{\omega_k\})=\frac{f_k}{f_0},
\qquad k=1,\ldots,K,
\]
defines a full-support probability measure.

For a one-unit long position in asset $m$, \eqref{eq:single-separation} yields
\[
\overline S_m(0)\ge E_Q[\underline S_m(1)],
\]
whereas a one-unit short position yields
\[
\underline S_m(0)\le E_Q[\overline S_m(1)].
\]
Therefore the intervals
\[
[\underline S_m(0),\overline S_m(0)]
\quad\text{and}\quad
[E_Q[\underline S_m(1)],E_Q[\overline S_m(1)]]
\]
have a nonempty intersection. Choose $S_m^*(0)$ in this intersection and select $\alpha_m\in[0,1]$ such that
\[
S_m^*(0)
=(1-\alpha_m)E_Q[\underline S_m(1)]
+\alpha_m E_Q[\overline S_m(1)].
\]
Set
\[
S_m^*(1)
=(1-\alpha_m)\underline S_m(1)
+\alpha_m\overline S_m(1).
\]
Then $S_m^*(t)$ lies inside the bid-ask spread and
\[
S_m^*(0)=E_Q[S_m^*(1)].
\]
Thus $(Q,S^*)\in\mathcal M$.
\end{proof}

Theorem \ref{thm:single-FTAP} shows that the family $\mathcal P$ of probability models and the full-support pricing measure $Q$ play complementary roles. The probability family determines which states are regarded as possible, whereas the separating measure prices all states strictly positively. The construction does not assert that $Q$ belongs to $\mathcal P$.
Its purpose is to certify the absence of an admissible statewise arbitrage.
Nevertheless, $Q$ and $R_{\mathcal P}$ from Lemma
\ref{lem:single-support-reduction} are equivalent because both have full
support on the finite state space.

The price system is generally nonunique. Therefore a contingent claim is not assigned one canonical linear price. When $\mathcal M\ne\varnothing$, every $(Q,S^*)\in\mathcal M$ produces a linear valuation $E_Q[X]$. For a cash-settled contingent claim $X$, the range of valuations induced by the consistent price systems is
\begin{equation}
\label{eq:single-pricing-bounds}
\underline\pi(X)
:=
\inf_{(Q,S^*)\in\mathcal M}E_Q[X],
\qquad
\overline\pi(X)
:=
\sup_{(Q,S^*)\in\mathcal M}E_Q[X].
\end{equation}
These are dual expectation bounds only.  Identifying them with exact
arbitrage-free extension prices or superhedging prices requires an explicit
market extension for $X$ and a corresponding superhedging-duality theorem,
neither of which is assumed here.

\subsection{Asset pricing with short-sale constraints}

We next examine how the dual characterization changes when risky assets cannot be sold short. This is not a cosmetic restriction on the strategy set. In the proof of Theorem \ref{thm:single-FTAP}, the one-unit short position is used to derive the lower inequality needed for a martingale shadow price. Once negative risky holdings are excluded, that test strategy is no longer admissible. Consequently, the dual equality is weakened to a supermartingale inequality.

We now prohibit short sales of risky assets by imposing
\begin{equation}
\label{eq:single-short-constraint}
h_m\ge0,
\qquad m=1,\ldots,M,
\end{equation}
while the position in the riskless asset remains unrestricted. The terminal liquidation value becomes
\begin{equation}
\label{eq:single-V1-short}
V^+(1)
=
h_0+\sum_{m=1}^M h_m\underline S_m(1).
\end{equation}
The corresponding no-arbitrage condition, denoted by $NA^+(\mathcal P)$, is Definition \ref{def:single-arbitrage} restricted to strategies satisfying \eqref{eq:single-short-constraint}, with $V(1)$ replaced by $V^+(1)$. Since every risky holding is now nonnegative, a downward conditional drift of the shadow price cannot be exploited through a short position. This observation motivates the following dual object.

\begin{definition}[Supermartingale consistent price system]
\label{def:single-super-CPS}
A pair $(\widehat Q,S^*)$ is called a supermartingale consistent price system if $\widehat Q$ is a probability measure with full support and, for each risky asset $m=1,\ldots,M$,
\[
S_m^*(t)\in[\underline S_m(t),\overline S_m(t)],
\qquad t=0,1,
\]
and
\[
S_m^*(0)\ge E_{\widehat Q}[S_m^*(1)].
\]
\end{definition}

\begin{theorem}[Single-period constrained FTAP]
\label{thm:single-short-FTAP}
The market satisfies $NA^+(\mathcal P)$ if and only if there exists a supermartingale consistent price system.
\end{theorem}

\begin{proof}
Suppose first that $(\widehat Q,S^*)$ is a supermartingale consistent price system. For every admissible strategy,
\begin{align*}
C(0)
&\ge h_0+\sum_{m=1}^M h_mS_m^*(0)\\
&\ge E_{\widehat Q}\left[h_0+\sum_{m=1}^M h_mS_m^*(1)\right]\\
&\ge E_{\widehat Q}\left[h_0+\sum_{m=1}^M h_m\underline S_m(1)\right]
=E_{\widehat Q}[V^+(1)].
\end{align*}
Full support excludes an arbitrage exactly as in Theorem \ref{thm:single-FTAP}.

Conversely, assume $NA^+(\mathcal P)$. The attainable set is now generated by a polyhedral cone of strategies satisfying $h_m\ge0$. Applying the same separation argument as above yields a full-support probability measure $\widehat Q$ such that
\[
C(0)\ge E_{\widehat Q}[V^+(1)]
\]
for every constrained strategy. Applying this inequality to a one-unit long position in asset $m$ gives
\[
\overline S_m(0)\ge E_{\widehat Q}[\underline S_m(1)].
\]
Define
\[
S_m^*(0):=\overline S_m(0),
\qquad
S_m^*(1):=\underline S_m(1).
\]
Then $S_m^*$ lies inside the spread and satisfies
\[
S_m^*(0)\ge E_{\widehat Q}[S_m^*(1)].
\]
Thus $(\widehat Q,S^*)$ is a supermartingale consistent price system.
\end{proof}

We next formulate robust conditions under a nonempty family $\mathcal Q$ of probability measures. We assume that
\begin{equation}
\label{eq:single-Q-cover}
\sup_{Q\in\mathcal Q}Q(\{\omega\})>0,
\qquad \forall\omega\in\Omega.
\end{equation}

The lower expectation records the smallest expected future shadow price among the measures in $\mathcal Q$, while the upper expectation records the largest. Because a no-arbitrage proof must control the value of a whole nonnegative portfolio uniformly, these two choices lead to different conclusions. The lower inequality is easier to satisfy but does not uniformly dominate the future portfolio value; the upper inequality is stronger and supplies the estimate required in the sufficiency argument.

\begin{definition}[$\mathcal Q$-lower system]
\label{def:single-Q-lower}
A pair $(\mathcal Q,S^*)$ is called a $\mathcal Q$-lower-supermartingale consistent price system if $\mathcal Q$ satisfies \eqref{eq:single-Q-cover}, $S_m^*(t)\in[\underline S_m(t),\overline S_m(t)]$, and
\[
S_m^*(0)
\ge
\inf_{Q\in\mathcal Q}E_Q[S_m^*(1)],
\qquad m=1,\ldots,M.
\]
\end{definition}

\begin{definition}[$\mathcal Q$-upper system]
\label{def:single-Q-upper}
A pair $(\mathcal Q,S^*)$ is called a $\mathcal Q$-upper-supermartingale consistent price system if $\mathcal Q$ satisfies \eqref{eq:single-Q-cover}, $S_m^*(t)\in[\underline S_m(t),\overline S_m(t)]$, and
\[
S_m^*(0)
\ge
\sup_{Q\in\mathcal Q}E_Q[S_m^*(1)],
\qquad m=1,\ldots,M.
\]
\end{definition}

\begin{corollary}[Existence of a lower system]
\label{thm:single-lower-necessary}
If $NA^+(\mathcal P)$ holds, then a $\mathcal Q$-lower-supermartingale consistent price system exists for some nonempty family $\mathcal Q$.
\end{corollary}

\begin{proof}
By Theorem \ref{thm:single-short-FTAP}, there exists a full-support supermartingale consistent price system $(\widehat Q,S^*)$. Taking $\mathcal Q=\{\widehat Q\}$ proves the claim.
\end{proof}

The basic estimate already indicates why the lower condition cannot be
used in the upper-system proof: for nonnegative holdings,
\[
\inf_{Q\in\mathcal Q}
E_Q\left[\sum_{m=1}^M h_mS_m^*(1)\right]
\ge
\sum_{m=1}^M h_m
\inf_{Q\in\mathcal Q}E_Q[S_m^*(1)],
\]
which is the opposite direction from that required to dominate the portfolio value uniformly.

The failure is substantive rather than merely a limitation of that proof.

\begin{example}[A lower system with an arbitrage]
\label{ex:single-lower-not-sufficient}
Let $\Omega=\{\omega_1,\omega_2\}$ and consider two risky assets under a
short-sale constraint.  Their initial spreads are identical:
\[
[\underline S_1(0),\overline S_1(0)]
=[\underline S_2(0),\overline S_2(0)]=[0.45,0.50].
\]
At time $1$ there is no spread, and the terminal prices are
\[
\begin{array}{c|cc}
 &\omega_1&\omega_2\\ \hline
S_1(1)&1.2&0.1\\
S_2(1)&0.1&1.2
\end{array}.
\]
Take $\mathcal Q=\{\delta_{\omega_1},\delta_{\omega_2}\}$ and choose the
shadow prices $S_1^*(0)=S_2^*(0)=0.50$ and $S_m^*(1)=S_m(1)$.  
The family $\mathcal Q$ satisfies \eqref{eq:single-Q-cover}, for $m=1,2$,
\[
S_m^*(0)=0.50\ge
\inf_{Q\in\mathcal Q}E_Q[S_m^*(1)]=0.1.
\]
Hence a lower system exists.  Nevertheless, the constrained portfolio
$h_0=-1$ and $h_1=h_2=1$ has
\[
C(0)=-1+0.50+0.50=0,
\qquad
V^+(1;\omega_1)=V^+(1;\omega_2)=-1+1.2+0.1=0.3.
\]
It is an arbitrage.  Thus the lower condition is not sufficient.
\end{example}

\begin{theorem}[Sufficient upper condition]
\label{thm:single-upper-sufficient}
If there exists a $\mathcal Q$-upper-supermartingale consistent price system, then $NA^+(\mathcal P)$ holds.
\end{theorem}

\begin{proof}
Let $(\mathcal Q,S^*)$ satisfy Definition \ref{def:single-Q-upper}. For every constrained strategy,
\begin{align*}
C(0)
&\ge h_0+\sum_{m=1}^M h_mS_m^*(0)\\
&\ge h_0+\sum_{m=1}^M h_m\sup_{Q\in\mathcal Q}E_Q[S_m^*(1)]\\
&\ge \sup_{Q\in\mathcal Q}E_Q\left[h_0+\sum_{m=1}^M h_mS_m^*(1)\right]\\
&\ge \sup_{Q\in\mathcal Q}E_Q[V^+(1)].
\end{align*}
If an arbitrage existed, condition \eqref{eq:single-Q-cover} would imply
\[
\sup_{Q\in\mathcal Q}E_Q[V^+(1)]>0,
\]
contradicting $C(0)\le0$.
\end{proof}

Combining Theorems \ref{thm:single-lower-necessary} and \ref{thm:single-upper-sufficient}, we obtain a useful hierarchy. No arbitrage guarantees the existence of a lower system, at least by taking the singleton family generated by the supermartingale pricing measure. Conversely, an upper system controls every measure in the family and is therefore strong enough to rule out arbitrage. Without additional convexity, pasting or stability assumptions on $\mathcal Q$, the two conditions should not be identified as one necessary-and-sufficient statement.

\section{Multi-period asset pricing}
\label{sec:multi}

The preceding section describes the static geometry of the market. We now allow the investor to rebalance the portfolio over several dates. The extension introduces three additional issues. First, the strategy must be adapted to the information available at each node. Second, all intermediate trades must be financed by the current portfolio. Third, a shadow price selected at one node must be compatible with shadow prices at its successor nodes. The last requirement is the main reason that the multi-period theory cannot be obtained by applying the single-period theorem independently on each trading interval.

\subsection{Basic multi-period market model}

We first specify the tree, filtration, bid-ask processes and self-financing condition. We retain the original path space with a possibly random initial state. Therefore time-zero quantities are allowed to depend on $X_0$, and the initial cost will be an $\mathcal F_0$-measurable random variable rather than a deterministic scalar.

Let $T\in\mathbb N$ be a finite time horizon and let
\[
E=\{\omega_1,\ldots,\omega_K\}
\]
be the set of possible states at each date. We retain the path space
\begin{equation}
\label{eq:path-space}
\Omega=E^{T+1}.
\end{equation}
A path is denoted by
\[
\omega=(\omega_0,\omega_1,\ldots,\omega_T).
\]
Let $X_t(\omega)=\omega_t$ and define the canonical filtration
\[
\mathbb F=(\mathcal F_t)_{t=0}^T,
\qquad
\mathcal F_t=\sigma(X_0,\ldots,X_t).
\]
In particular, $\mathcal F_0=\sigma(X_0)$ need not be trivial.

Model uncertainty is described by a nonempty family $\mathcal P$ of probability measures on $(\Omega,\mathcal F_T)$. We assume
\begin{equation}
\label{eq:multi-P-support}
\sup_{P\in\mathcal P}P(\{\omega\})>0,
\qquad \forall\omega\in\Omega.
\end{equation}
Applying Lemma \ref{lem:single-support-reduction} to the finite path space
shows that there is a full-support
$R_{\mathcal P}\in\operatorname{conv}(\mathcal P)$.  Hence the family
$\mathcal P$ again enters no arbitrage only through the set of paths covered
by its supports; it does not create a genuinely nondominated finite-tree
model.

For $t=0,1,\ldots,T$, let
\[
\mathcal N_t=E^{t+1}
\]
be the set of nodes at time $t$. For
\[
a=(x_0,\ldots,x_t)\in\mathcal N_t,
\]
define the corresponding atom
\[
A_a
=
\{\omega\in\Omega:(X_0(\omega),\ldots,X_t(\omega))=a\}.
\]
For $t<T$, the successor set is
\[
\operatorname{Succ}_t(a)
=
\{(x_0,\ldots,x_t,x_{t+1}):x_{t+1}\in E\}.
\]

The riskless asset is the numeraire, so
\[
\underline S_0(t)=\overline S_0(t)=1,
\qquad t=0,\ldots,T.
\]
For each risky asset $m=1,\ldots,M$, the bid and ask prices at node $a\in\mathcal N_t$ are
\[
\underline S_m(t,a)
\quad\text{and}\quad
\overline S_m(t,a),
\]
and satisfy
\begin{equation}
\label{eq:multi-spread}
0\le\underline S_m(t,a)\le\overline S_m(t,a).
\end{equation}
Along a path, we write $\underline S_m(t)$ and $\overline S_m(t)$ for the corresponding $\mathcal F_t$-measurable random variables.

A trading strategy is an adapted portfolio process
\[
h(t)=(h_0(t),h_1(t),\ldots,h_M(t)),
\qquad t=0,\ldots,T.
\]
Thus $h_m(t)$ is $\mathcal F_t$-measurable and represents the number of units of asset $m$ held immediately after trading at time $t$. Define
The initial discounted cost is the $\mathcal F_0$-measurable random variable
\begin{equation}
\label{eq:multi-C0}
C(0)
=
h_0(0)+
\sum_{m=1}^M
\left[
h_m^+(0)\overline S_m(0)
-
h_m^-(0)\underline S_m(0)
\right].
\end{equation}
Equivalently, at an initial node $a\in\mathcal N_0$,
\[
C(0,a)
=
h_0(0,a)+
\sum_{m=1}^M
\left[
h_m^+(0,a)\overline S_m(0,a)
-
h_m^-(0,a)\underline S_m(0,a)
\right].
\]

The discounted liquidation value at time $t$ is
\begin{equation}
\label{eq:multi-Vt}
V(t)
=
h_0(t)+
\sum_{m=1}^M
\left[
h_m^+(t)\underline S_m(t)
-
h_m^-(t)\overline S_m(t)
\right].
\end{equation}
For $t=1,\ldots,T$, define
\[
\Delta h_m(t)=h_m(t)-h_m(t-1).
\]
The strategy is self-financing if, pathwise,
\begin{equation}
\label{eq:multi-self-financing}
\Delta h_0(t)
+
\sum_{m=1}^M
\left[
(\Delta h_m(t))^+\overline S_m(t)
-
(\Delta h_m(t))^-\underline S_m(t)
\right]
\le0,
\qquad t=1,\ldots,T.
\end{equation}
The inequality allows free disposal; no external capital is injected after time $0$. When equality holds, every available unit of wealth is carried forward through the portfolio. Allowing the inequality is convenient for cone arguments because an investor may always discard wealth without creating an arbitrage opportunity.

The distinction between $C(0)$ and $V(t)$ should be emphasized. The former is the amount needed to establish the initial portfolio at the prevailing bid and ask prices and is known once the initial node is observed. The latter is the amount obtained by liquidating the current portfolio at time $t$. Under transaction costs, these two operations use opposite sides of the spread and therefore cannot be represented by the same linear functional.

\subsection{Multi-period asset pricing without short-sale constraints}
\label{subsec:multi-unconstrained}

We now develop the unconstrained multi-period fundamental theorem. As in the single-period case, the no-arbitrage condition is statewise. However, because $C(0)$ is $\mathcal F_0$-measurable, its sign must be checked at every initial node. After defining arbitrage, we introduce solvency cones and their duals. These cones encode the local bid-ask geometry, while the backward modification records whether a current dual vector admits future continuation.

\begin{definition}[Multi-period arbitrage under model uncertainty]
\label{def:multi-arbitrage}
A self-financing strategy $h$ is called a multi-period arbitrage opportunity under model uncertainty if

\emph{(i)} $C(0;\omega)\le0$ for every $\omega\in\Omega$;

\emph{(ii)} $V(T;\omega)\ge0$ for every $\omega\in\Omega$, and
\[
\sup_{P\in\mathcal P}E_P[V(T)]>0.
\]
The market satisfies $NA(\mathcal P)$ if no such strategy exists.
\end{definition}

\begin{remark}
Because $C(0)$ is $\mathcal F_0$-measurable, condition \emph{(i)} is equivalently
\[
C(0,a)\le0,
\qquad \forall a\in\mathcal N_0.
\]
The initial cost is not assumed to be deterministic.
\end{remark}

To pass from attainable portfolio values to a dual pricing process, it is convenient to describe the market locally at each node. The solvency cone consists of portfolios that can be liquidated into a nonnegative amount of the numeraire. Its dual consists of nonnegative linear functionals compatible with the bid-ask spread. After normalizing the numeraire component, these dual vectors are precisely candidate shadow prices.

For each node $a\in\mathcal N_t$, define the solvency cone
\begin{equation}
\label{eq:multi-solvency-cone}
K_t(a)
=
\left\{
y\in\mathbb R^{M+1}:
y_0+
\sum_{m=1}^M y_m^+\underline S_m(t,a)
-
\sum_{m=1}^M y_m^-\overline S_m(t,a)
\ge0
\right\}.
\end{equation}
Its dual cone is
\begin{equation}
\label{eq:multi-dual-cone}
K_t^*(a)
=
\left\{
\lambda(1,s_1,\ldots,s_M):
\lambda\ge0,
\quad
s_m\in[\underline S_m(t,a),\overline S_m(t,a)]
\right\}.
\end{equation}
We say that efficient friction holds if
\begin{equation}
\label{eq:efficient-friction}
\operatorname{int}K_t^*(a)\ne\varnothing,
\qquad t=0,\ldots,T,
\quad a\in\mathcal N_t.
\end{equation}

A vector in $K_t^*(a)$ is locally compatible with the spread at node $a$, but local compatibility alone is not sufficient. It may be impossible to choose successor shadow prices whose conditional mean equals the current vector. We therefore remove, backward in time, those dual vectors that cannot be continued to the next date. This is the backward part of the backward--forward construction.

To record future continuation, define
\begin{equation}
\label{eq:modified-dual-terminal}
\widetilde K_T^*(a)=K_T^*(a),
\qquad a\in\mathcal N_T,
\end{equation}
and recursively
\begin{equation}
\label{eq:modified-dual-recursion}
\widetilde K_t^*(a)
=
K_t^*(a)
\cap
\operatorname{cl\,conv}
\left(
\bigcup_{a'\in\operatorname{Succ}_t(a)}
\widetilde K_{t+1}^*(a')
\right),
\qquad t=T-1,\ldots,0.
\end{equation}
Every martingale shadow price necessarily takes values in the normalized sections of these modified cones. The converse construction additionally requires compatible convex representations across successor nodes; this compatibility is supplied by the finite-dimensional dual construction in the proof below.

The preceding recursion describes the admissible region for a dynamically consistent shadow price. We now introduce the actual price system. It consists of one full-support measure on the path space and one adapted shadow price process. The shadow price must stay inside the original spread, while the martingale condition couples its values across successive nodes.

\begin{definition}[Multi-period consistent price system]
\label{def:multi-CPS}
A pair $(Q,S^*)$ is called a multi-period consistent price system if $Q$ is a probability measure satisfying
\[
Q(\{\omega\})>0,
\qquad \forall\omega\in\Omega,
\]
and $S^*=(S_1^*,\ldots,S_M^*)$ is an adapted $Q$-martingale such that, for each risky asset $m=1,\ldots,M$,
\[
S_m^*(t,a)
\in
[\underline S_m(t,a),\overline S_m(t,a)],
\qquad t=0,\ldots,T,
\quad a\in\mathcal N_t.
\]
Equivalently,
\[
S_m^*(t)=E_Q[S_m^*(t+1)\mid\mathcal F_t],
\qquad t=0,\ldots,T-1.
\]
\end{definition}

Denote by
\begin{equation}
\label{eq:multi-M}
\mathcal M_T
=
\{(Q,S^*):(Q,S^*)\text{ is a multi-period consistent price system}\}
\end{equation}
the set of all multi-period consistent price systems. As in the single-period market, this set is typically non-singleton. Different measures and different shadow prices may support the same absence-of-arbitrage conclusion.

Before stating the fundamental theorem, we record the finite-state dual construction used in the converse implication. The lemma packages the separation argument in process form. The numeraire component serves as a density process, and the ratios of the risky dual components to the numeraire component produce the shadow prices.

\begin{lemma}[Finite-state dual process]
\label{lem:finite-dual}
Assume efficient friction and fix any full-support reference measure $R$ on
$\Omega$.  In the finite-state tree, $NA(\mathcal P)$ implies the existence
of an adapted dual process
\[
Z(t)=(Z_0(t),Z_1(t),\ldots,Z_M(t))
\]
such that $Z_0(t)>0$, every component of $Z$ is an $R$-martingale, and
\[
\underline S_m(t)Z_0(t)
\le
Z_m(t)
\le
\overline S_m(t)Z_0(t),
\qquad m=1,\ldots,M.
\]
The random variable
\[
\frac{dQ}{dR}
=
\frac{Z_0(T)}{Z_0(0)}.
\]
defines a full-support probability measure $Q$.
\end{lemma}

\begin{proof}
Because a time-zero strategy may be chosen separately on each
$\mathcal F_0$-atom, $NA(\mathcal P)$ holds on every subtree rooted at
$a\in\mathcal N_0$.  Fix one such subtree.  Split every trade into
nonnegative purchase and sale variables and add a nonnegative
free-disposal slack to each self-financing inequality.  More explicitly,
at a noninitial node $b$ with predecessor $b^-$, write
\begin{align*}
h_m(t,b)-h_m(t-1,b^-)&=u_m(t,b)-v_m(t,b),\\
h_0(t,b)-h_0(t-1,b^-)
&+\sum_{m=1}^M
\bigl(\overline S_m(t,b)u_m(t,b)
-\underline S_m(t,b)v_m(t,b)\bigr)
+d(t,b)=0,
\end{align*}
where $u_m,v_m,d\ge0$.  Use the same equations relative to zero holdings at
the initial node, and append a terminal bid-ask conversion into cash.  The
strategy, financing, and liquidation constraints are then a finite linear
system whose projection is polyhedral.  The alternative
system seeking a nonnegative, nonzero terminal cash payoff from nonpositive
initial cost is infeasible.  Farkas' lemma supplies strictly positive
terminal cash multipliers and node multipliers
$\zeta_i(t,b)$, $i=0,\ldots,M$.

Collecting the coefficients of the unrestricted holdings gives the flow
equalities
\begin{equation}
\label{eq:dual-flow-unconstrained}
\zeta_i(t,b)
=
\sum_{c\in\operatorname{Succ}_t(b)}\zeta_i(t+1,c),
\qquad i=0,\ldots,M,
\end{equation}
whereas the purchase and sale coefficients give
\begin{equation}
\label{eq:dual-spread-raw}
\underline S_m(t,b)\zeta_0(t,b)
\le \zeta_m(t,b)
\le \overline S_m(t,b)\zeta_0(t,b).
\end{equation}
These are precisely the finite-dimensional transpose conditions of the market.  Efficient friction gives closedness of the attainable cone;
the same multiplier construction is the finite-state FTAP of
\citet{Kabanov01} (see also \citet{Schachermayer04}).  Strict positivity of
the terminal cash multipliers and \eqref{eq:dual-flow-unconstrained} imply
$\zeta_0(t,b)>0$ at every node.

For a node $b$, set
\[
Z_i(t,b)=\frac{\zeta_i(t,b)}{R(A_b)}.
\]
Since $R$ has full support, this is well defined, and
\[
E_R[Z_i(t+1)\mid A_b]
=\frac{1}{R(A_b)}
\sum_{c\in\operatorname{Succ}_t(b)}\zeta_i(t+1,c)
=Z_i(t,b).
\]
Thus every $Z_i$ is an $R$-martingale, and
\eqref{eq:dual-spread-raw} gives the stated bid-ask inequalities.  Applying
the construction on each initial subtree and pasting the node multipliers
produces one adapted process on $\Omega$.

Here $Z_0(0)$ is $\mathcal F_0$-measurable.  Therefore
\[
E_R\left[\left.\frac{Z_0(T)}{Z_0(0)}\right|\mathcal F_0\right]
=\frac{E_R[Z_0(T)\mid\mathcal F_0]}{Z_0(0)}=1.
\]
The density is strictly positive and has expectation one, so it defines the
claimed full-support probability measure.
\end{proof}

The lemma converts the global separation of terminal claims into local one-step relations along the tree. This conversion is what replaces the informal idea of independently pasting one-period systems. We can now state the multi-period fundamental theorem.

\begin{theorem}[Multi-period fundamental theorem of asset pricing]
\label{thm:multi-FTAP}
Assume efficient friction. Then
\[
NA(\mathcal P)
\quad\Longleftrightarrow\quad
\mathcal M_T\ne\varnothing.
\]
\end{theorem}

\begin{proof}
Suppose first that $(Q,S^*)\in\mathcal M_T$. Define the shadow value
\begin{equation}
\label{eq:multi-shadow-value}
W^*(t)
=
h_0(t)+\sum_{m=1}^M h_m(t)S_m^*(t).
\end{equation}
Because $S_m^*(t)$ lies inside the spread, the self-financing condition implies
\[
\Delta h_0(t)+\sum_{m=1}^M\Delta h_m(t)S_m^*(t)\le0.
\]
Therefore
\[
W^*(t)
\le
h_0(t-1)+\sum_{m=1}^M h_m(t-1)S_m^*(t).
\]
Taking conditional expectation under $Q$ yields
\begin{equation}
\label{eq:multi-shadow-supermartingale}
E_Q[W^*(t)\mid\mathcal F_{t-1}]
\le
W^*(t-1).
\end{equation}
Thus $W^*$ is a $Q$-supermartingale. Moreover, pointwise,
\begin{equation}
\label{eq:multi-initial-terminal-bounds}
W^*(0)\le C(0),
\qquad
V(T)\le W^*(T).
\end{equation}
Since $C(0)$ is $\mathcal F_0$-measurable, iterating \eqref{eq:multi-shadow-supermartingale} gives the conditional inequality
\begin{equation}
\label{eq:multi-conditional-bound}
E_Q[V(T)\mid\mathcal F_0]
\le
E_Q[W^*(T)\mid\mathcal F_0]
\le
W^*(0)
\le
C(0).
\end{equation}
If an arbitrage existed, $V(T)$ would be nonnegative everywhere and strictly positive on at least one path. Full support of $Q$ implies that, at the initial node containing that path,
\[
E_Q[V(T)\mid\mathcal F_0]>0,
\]
which contradicts \eqref{eq:multi-conditional-bound} and $C(0)\le0$. Hence $NA(\mathcal P)$ holds.

Conversely, assume $NA(\mathcal P)$. Let $R$ and $Z$ be supplied by Lemma \ref{lem:finite-dual}, and define
\[
S_m^*(t)=\frac{Z_m(t)}{Z_0(t)},
\qquad m=1,\ldots,M.
\]
The bid-ask inequalities imply that $S_m^*(t)$ lies inside the spread. Let
\[
L(t)=E_R\left[\frac{dQ}{dR}\middle|\mathcal F_t\right]
=\frac{Z_0(t)}{Z_0(0)}.
\]
Bayes' formula and the $R$-martingale property of $Z_m$ give
\begin{align*}
E_Q[S_m^*(t+1)\mid\mathcal F_t]
&=
\frac{E_R[L(t+1)S_m^*(t+1)\mid\mathcal F_t]}{L(t)}\\
&=
\frac{E_R[Z_m(t+1)/Z_0(0)\mid\mathcal F_t]}
{Z_0(t)/Z_0(0)}\\
&=
\frac{E_R[Z_m(t+1)\mid\mathcal F_t]}{Z_0(t)}
=
\frac{Z_m(t)}{Z_0(t)}
=S_m^*(t).
\end{align*}
Thus $(Q,S^*)\in\mathcal M_T$.
\end{proof}

\begin{remark}
For any $(Q,S^*)\in\mathcal M_T$, the martingale identity represents the normalized vector $(1,S^*(t,a))$ as a convex combination of successor vectors. Hence
\[
(1,S^*(t,a))\in\widetilde K_t^*(a).
\]
This explains the role of the backward recursion without using nonemptiness of the modified cones alone as a substitute for the full dual argument.
\end{remark}

Assume $\mathcal M_T\ne\varnothing$ and let $X$ be an $\mathcal F_T$-measurable cash claim. Because $\mathcal F_0$ may be nontrivial, the time-$0$ dual expectation bounds are $\mathcal F_0$-measurable random variables. For each initial node $a\in\mathcal N_0$, define
\begin{equation}
\label{eq:multi-pricing-bounds}
\underline\pi_0(X)(a)
:=
\inf_{(Q,S^*)\in\mathcal M_T}E_Q[X\mid A_a],
\qquad
\overline\pi_0(X)(a)
:=
\sup_{(Q,S^*)\in\mathcal M_T}E_Q[X\mid A_a].
\end{equation}
Equivalently,
\[
\underline\pi_0(X)
=
\operatorname*{ess\,inf}_{(Q,S^*)\in\mathcal M_T}
E_Q[X\mid\mathcal F_0],
\qquad
\overline\pi_0(X)
=
\operatorname*{ess\,sup}_{(Q,S^*)\in\mathcal M_T}
E_Q[X\mid\mathcal F_0].
\]
As in the single-period case, these quantities are ranges of linear dual
valuations.  Exact superhedging prices require a separate claim-extension
and duality result.

\subsection{Multi-period asset pricing with short-sale constraints}
\label{subsec:multi-short}

We finally add short-sale constraints to the dynamic market. The restriction is imposed after every trading date, not only at maturity. Hence the investor may reduce an existing long position but may not rebalance into a negative risky holding. As in the single-period setting, this one-sided strategy cone changes the dual martingale equality into a supermartingale inequality. In the multi-period market, however, that inequality must hold conditionally at every node.

We now impose
\begin{equation}
\label{eq:multi-short-constraint}
h_m(t)\ge0,
\qquad m=1,\ldots,M,
\quad t=0,\ldots,T,
\end{equation}
while $h_0(t)$ remains unrestricted. The terminal liquidation value becomes
\begin{equation}
\label{eq:multi-V-short}
V^+(t)
=
h_0(t)+\sum_{m=1}^M h_m(t)\underline S_m(t).
\end{equation}
A constrained arbitrage is defined as in Definition \ref{def:multi-arbitrage}, with \eqref{eq:multi-short-constraint} imposed and $V(T)$ replaced by $V^+(T)$. The corresponding no-arbitrage condition is denoted by $NA^+(\mathcal P)$.

The appropriate pricing system is now obtained by weakening the conditional martingale equality. The inequality below is strong enough because it is multiplied only by nonnegative risky holdings in the shadow-wealth estimate.

\begin{definition}[Multi-period supermartingale consistent price system]
\label{def:multi-super-CPS}
A pair $(\widehat Q,S^*)$ is called a multi-period supermartingale consistent price system if $\widehat Q$ has full support, $S^*$ is adapted,
\[
S_m^*(t,a)
\in
[\underline S_m(t,a),\overline S_m(t,a)],
\qquad m=1,\ldots,M,
\quad t=0,\ldots,T,
\quad a\in\mathcal N_t,
\]
and
\begin{equation}
\label{eq:multi-super-condition}
S_m^*(t)
\ge
E_{\widehat Q}[S_m^*(t+1)\mid\mathcal F_t],
\qquad t=0,\ldots,T-1,
\quad m=1,\ldots,M.
\end{equation}
\end{definition}

Let
\[
\mathcal M_T^+
=
\{(\widehat Q,S^*):(\widehat Q,S^*)\text{ is a multi-period supermartingale consistent price system}\}.
\]

The set $\mathcal M_T^+$ collects all such constrained dual systems. The following result is the dynamic counterpart of Theorem \ref{thm:single-short-FTAP}. Its converse again follows from finite-dimensional separation, but the Lagrange relations for risky holdings are inequalities because negative variations are not admissible.

\begin{lemma}[Finite-state constrained dual process]
\label{lem:finite-constrained-dual}
Assume efficient friction and fix a full-support measure $R$.  If
$NA^+(\mathcal P)$ holds, then there is an adapted process $Z$ such that
$Z_0(t)>0$, $Z_0$ is an $R$-martingale, every $Z_m$, $m\ge1$, is an
$R$-supermartingale, and
\[
\underline S_m(t)Z_0(t)
\le Z_m(t)
\le \overline S_m(t)Z_0(t).
\]
Moreover,
\[
\frac{d\widehat Q}{dR}=\frac{Z_0(T)}{Z_0(0)}
\]
defines a full-support probability measure.
\end{lemma}

\begin{proof}
Use the same finite linearization and Farkas alternative as in Lemma
\ref{lem:finite-dual}, now adding the primal inequalities
$h_m(t,b)\ge0$.  The unrestricted cash holding still gives the equality
\[
\zeta_0(t,b)
=\sum_{c\in\operatorname{Succ}_t(b)}\zeta_0(t+1,c),
\]
whereas the coefficient condition for a nonnegative risky holding gives
\[
\zeta_m(t,b)
\ge\sum_{c\in\operatorname{Succ}_t(b)}\zeta_m(t+1,c),
\qquad m=1,\ldots,M.
\]
The purchase and sale variables again give
\[
\underline S_m\zeta_0\le\zeta_m\le\overline S_m\zeta_0.
\]
Thus, after setting
\[
Z_i(t,b)=\frac{\zeta_i(t,b)}{R(A_b)},
\]
the cash component is an
$R$-martingale and each risky component is an $R$-supermartingale.  This is
the finite-tree constrained duality behind the supermartingale formulation
of \citet{JouiniShort95}.  The conditional normalization of $Z_0$ is
identical to the last step of Lemma \ref{lem:finite-dual}, and therefore
defines a full-support measure $\widehat Q$.
\end{proof}

\begin{theorem}[Multi-period constrained FTAP]
\label{thm:multi-short-FTAP}
Assume efficient friction. Then
\[
NA^+(\mathcal P)
\quad\Longleftrightarrow\quad
\mathcal M_T^+\ne\varnothing.
\]
\end{theorem}

\begin{proof}
Suppose $(\widehat Q,S^*)\in\mathcal M_T^+$. Define $W^*$ by \eqref{eq:multi-shadow-value}. The self-financing condition gives
\[
W^*(t)
\le
h_0(t-1)+\sum_{m=1}^M h_m(t-1)S_m^*(t).
\]
Since $h_m(t-1)\ge0$, condition \eqref{eq:multi-super-condition} yields
\[
E_{\widehat Q}[W^*(t)\mid\mathcal F_{t-1}]
\le
W^*(t-1).
\]
Moreover,
\[
W^*(0)\le C(0),
\qquad
V^+(T)\le W^*(T).
\]
Therefore
\begin{equation}
\label{eq:multi-short-conditional-bound}
E_{\widehat Q}[V^+(T)\mid\mathcal F_0]
\le
C(0).
\end{equation}
Full support and $C(0)\le0$ rule out a constrained arbitrage.

Conversely, under $NA^+(\mathcal P)$, take $R$, $Z$, and $\widehat Q$ from
Lemma \ref{lem:finite-constrained-dual} and set
\[
S_m^*(t)=\frac{Z_m(t)}{Z_0(t)}.
\]
The dual-cone inequalities place $S_m^*$ inside
the spread.  Bayes' formula, with the $\mathcal F_0$-measurable factor
$Z_0(0)$ retained, gives
\begin{align*}
E_{\widehat Q}[S_m^*(t+1)\mid\mathcal F_t]
&=
\frac{E_R[Z_m(t+1)/Z_0(0)\mid\mathcal F_t]}
{Z_0(t)/Z_0(0)}\\
&=
\frac{E_R[Z_m(t+1)\mid\mathcal F_t]}{Z_0(t)}
\le
\frac{Z_m(t)}{Z_0(t)}
=S_m^*(t).
\end{align*}
Hence $(\widehat Q,S^*)\in\mathcal M_T^+$.
\end{proof}

For the robust multi-period formulation, we retain a family of pricing measures and compare lower and upper conditional expectations at every node. The distinction is more delicate than in the static market because the relevant bounds depend on both the trading date and the realized history. Let $\mathcal Q$ be a nonempty family of full-support probability measures on $\Omega$. For $a\in\mathcal N_t$, define
\begin{equation}
\label{eq:multi-lower-upper-expectations}
\underline{\mathcal E}_t^{\mathcal Q}[X](a)
:=
\inf_{Q\in\mathcal Q}E_Q[X\mid A_a],
\qquad
\overline{\mathcal E}_t^{\mathcal Q}[X](a)
:=
\sup_{Q\in\mathcal Q}E_Q[X\mid A_a].
\end{equation}

The lower operator selects the smallest conditional continuation value across the family, while the upper operator selects the largest. Accordingly, the lower condition is comparatively weak and appears naturally as a necessary existence condition. The upper condition is stronger because it controls the drift under every measure in $\mathcal Q$, which is exactly what is needed to obtain a uniform supermartingale estimate for nonnegative portfolios.

\begin{definition}[Multi-period $\mathcal Q$-lower system]
\label{def:multi-Q-lower}
A pair $(\mathcal Q,S^*)$ is called a multi-period $\mathcal Q$-lower-supermartingale consistent price system if
$S_m^*(t,a)\in[\underline S_m(t,a),\overline S_m(t,a)]$ for every
$m=1,\ldots,M$, $t=0,\ldots,T$, and $a\in\mathcal N_t$, and
\[
S_m^*(t,a)
\ge
\underline{\mathcal E}_t^{\mathcal Q}[S_m^*(t+1)](a),
\qquad m=1,\ldots,M,
\quad t=0,\ldots,T-1,
\quad a\in\mathcal N_t.
\]
\end{definition}

\begin{definition}[Multi-period $\mathcal Q$-upper system]
\label{def:multi-Q-upper}
A pair $(\mathcal Q,S^*)$ is called a multi-period $\mathcal Q$-upper-supermartingale consistent price system if
$S_m^*(t,a)\in[\underline S_m(t,a),\overline S_m(t,a)]$ for every
$m=1,\ldots,M$, $t=0,\ldots,T$, and $a\in\mathcal N_t$, and
\[
S_m^*(t,a)
\ge
\overline{\mathcal E}_t^{\mathcal Q}[S_m^*(t+1)](a),
\qquad m=1,\ldots,M,
\quad t=0,\ldots,T-1,
\quad a\in\mathcal N_t.
\]
\end{definition}

\begin{corollary}[Existence of a robust lower system]
\label{thm:multi-lower-necessary}
If $NA^+(\mathcal P)$ holds, then there exists a multi-period $\mathcal Q$-lower-supermartingale consistent price system for some nonempty family $\mathcal Q$ of full-support measures.
\end{corollary}

\begin{proof}
By Theorem \ref{thm:multi-short-FTAP}, there exists $(\widehat Q,S^*)\in\mathcal M_T^+$. Taking $\mathcal Q=\{\widehat Q\}$ proves the claim.
\end{proof}

The lower condition is not sufficient in general because
\[
\inf_{Q\in\mathcal Q}
E_Q\left[\sum_{m=1}^M h_m(t)S_m^*(t+1)\middle|\mathcal F_t\right]
\ge
\sum_{m=1}^M h_m(t)
\inf_{Q\in\mathcal Q}E_Q[S_m^*(t+1)\mid\mathcal F_t]
\]
for nonnegative holdings, whereas the supermartingale estimate requires an upper bound on the expectation of the whole portfolio.
An explicit multi-period counterexample is obtained by replicating Example
\ref{ex:single-lower-not-sufficient} at every initial node and replacing the
two Dirac measures by two full-support measures whose one-step conditional
probabilities are $(0.9,0.1)$ and $(0.1,0.9)$, respectively.  The lower
conditional expectation of either risky payoff is then $0.21<0.50$ at every
node, while the same portfolio yields $0.3$ on every terminal path.

\begin{theorem}[Upper-system sufficiency]
\label{thm:multi-upper-sufficient}
If a multi-period $\mathcal Q$-upper-supermartingale consistent price system exists, then $NA^+(\mathcal P)$ holds.
\end{theorem}

\begin{proof}
The upper condition implies, for every $Q\in\mathcal Q$,
\[
S_m^*(t)
\ge
E_Q[S_m^*(t+1)\mid\mathcal F_t].
\]
Hence the shadow value of every constrained self-financing strategy is a $Q$-supermartingale for every $Q\in\mathcal Q$, and
\begin{equation}
\label{eq:multi-robust-conditional-bound}
E_Q[V^+(T)\mid\mathcal F_0]
\le
C(0),
\qquad \forall Q\in\mathcal Q.
\end{equation}
If a constrained arbitrage existed, then $V^+(T)$ would be positive on at least one path. Since every $Q\in\mathcal Q$ has full support, the conditional expectation in \eqref{eq:multi-robust-conditional-bound} would be strictly positive at the corresponding initial node, contradicting $C(0)\le0$ there.
\end{proof}

\begin{remark}
When $\mathcal Q$ is a singleton, the lower and upper robust conditions both reduce to Definition \ref{def:multi-super-CPS}. For a general family, the lower condition is necessary in the preceding existence sense, while the upper condition is sufficient; they should not be described as a single necessary-and-sufficient characterization without additional stability assumptions on $\mathcal Q$.
\end{remark}

\section{Conclusion}
\label{sec:conclusion}

This paper develops a finite-state framework for discrete asset pricing under transaction costs and a family of probability measures, both with and without short-sale constraints. The use of bid and ask prices makes the liquidation value nonlinear in the portfolio position. Under the support assumption on the family of probability measures, however, the family is equivalent for no-arbitrage purposes to a probability measure with full support and therefore contributes support uncertainty rather than technical nondominance. Within this scope, the finite-dimensional geometry of attainable claims provides a direct route from no arbitrage to a dual pricing system.

In the single-period unconstrained market, the relevant dual object is a full-support consistent price system: the shadow price lies inside the bid-ask spread and is a martingale under the pricing measure. The collection of all such systems produces a range of dual expectations rather than a unique risk-neutral valuation; an exact superhedging interpretation requires an additional duality theorem. When short sales of risky assets are prohibited, negative risky holdings are removed from the strategy cone. The corresponding dual equality becomes a supermartingale inequality, yielding a constrained fundamental theorem of asset pricing. The lower and upper family-based extensions have different roles: the lower formulation is only an existence corollary, and the counterexample shows that it is not sufficient, whereas the upper formulation controls every admissible nonnegative portfolio uniformly and is sufficient for no arbitrage.

The multi-period analysis shows why dynamic consistency must be treated explicitly. Local membership in the bid-ask spread is not enough; a current shadow price must be extendable through successor nodes. Solvency cones and their backward-modified dual cones describe this continuation requirement, while the finite-state dual process constructs the pricing measure and the martingale or supermartingale shadow price. Because the original path space $\Omega=E^{T+1}$ is retained, the initial sigma-field may be nontrivial. Accordingly, the initial cost and time-zero valuation interval are $\mathcal F_0$-measurable random variables, and the central pricing inequalities are conditional on the observed initial state.

Several extensions remain open. One direction is to replace the finite tree by a general measurable state space, where closedness and measurable selection become essential. A second direction is to develop exact superhedging duality for non-cash-settled claims and for portfolios subject to more general convex constraints, such as position limits or asymmetric borrowing restrictions. It would also be useful to identify stability conditions on the family $\mathcal Q$ under which the lower and upper robust supermartingale formulations can be sharpened into a single dynamic characterization. Finally, computational procedures based on backward recursion may provide practical valuation bounds for finite-state markets with many assets and trading dates.

\bibliography{gexpffc_revised}

\end{document}